\documentclass{amsart}

\makeatletter
 \def\LaTeX{\leavevmode L\raise.42ex
   \hbox{\kern-.3em\size{\sf@size}{0pt}\selectfont A}\kern-.15em\TeX}
\makeatother

\newcommand{\BibTeX}{{\rm B\kern-.05em{\sc
i\kern-.025emb}\kern-.08em\TeX}}
\newtheorem{col}{Corollary}[section]
\newtheorem{thm}{Theorem}[section]
\newtheorem{lem}[thm]{Lemma}
\newtheorem{rem}{Remark}
\theoremstyle{definition}
\newtheorem{defn}{Definition}
\newtheorem{exmp}{Example}

\numberwithin{equation}{section}

\begin{document}

\title[Sampling, filtering and approximation on
 graphs] {Sampling, filtering  and sparse
approximations  on combinatorial graphs}

\maketitle

\begin{center}

\author{Isaac Z. Pesenson }\footnote{ Department of Mathematics, Temple University,
 Philadelphia,
PA 19122; pesenson@math.temple.edu. The author was supported in
part by the National Geospatial-Intelligence Agency University
Research Initiative (NURI), grant HM1582-08-1-0019. }

\author{Meyer Z. Pesenson }\footnote{ Spitzer Science Center, California Institute of
Technology, MC 314-6, Pasadena, CA 91125;  misha@ipac.caltech.edu.
The author was supported in part by the National
Geospatial-Intelligence Agency University Research Initiative
(NURI), grant HM1582-08-1-0019.}
\end{center}

 \begin{abstract}
In this paper we address sampling and approximation of functions
  on combinatorial graphs. We develop filtering on graphs by using
  Schrodinger's group of operators generated by combinatorial
  Laplace operator.
   Then we construct  a sampling theory by proving
   Poincare and  Plancherel-Polya-type inequalities for functions
   on graphs. These
  results lead to
      a theory of sparse approximations on graphs and
      have potential applications to  filtering, denoising,
      data dimension reduction,
  image processing, image compression, computer graphics, visualization and learning theory.

\end{abstract}

 {\bf Keywords and phrases: Combinatorial Laplace operator,
Poincare and Plancherel-Polya inequalities, Paley-Wiener spaces,
best approximations, sparse approximations, Schr\"{o}dinger
Semigroup, modulus of continuity, Hilbert frames.}

 {\bf Subject classifications:}
{Primary: 42C99, 05C99, 94A20; Secondary: 94A12 }

 \section{Introduction}

During the last years harmonic analysis on combinatorial graphs
attracted considerable attention. The interest  is stimulated in
part by multiple existing and potential  applications of analysis
on graphs to information theory, signal analysis, image
processing, computer sciences, learning theory, astronomy
\cite{BN}, \cite{BMN},  \cite{BCMS}--\cite{CM2}, \cite{HK},
\cite{PWR}, \cite{MP1}-- \cite{RWP}.

Some of the approaches to  large data sets or images consider them
as graphs. However, for hyperspectral images, for example, this
 leads to graphs with too many vertices imbedded into  high  dimensional
  spaces, thus making dimension reduction necessary for effective data
mining.

It seems that one  possible way to approach this problem is by
using ideas from the classical sampling theory which has already
proved  very fruitful in various  branches  of applied
mathematics.

Let us remind  the  Classical Shannon-Nyquist sampling Theorem. It
states that for all Paley-Wiener functions of a fixed bandwidth
defined on Euclidean space one can find "not  very dense" sampling
sets which can be used to represent all relevant Paley-Wiener
functions. In some sense it allows to reduce the set of all points
of Euclidean space to a countable set of points. Moreover, since
the set of all Paley-Wiener functions is dense in the space
$L_{2}(\mathbb{R}^{d})$ one can use this property to approach
sampling of non-Paley-Wiener functions.

The goal of this work, to show that analysis of lower frequencies
on a graph can be performed on a smaller subgraph. Note that in
many situations lower frequencies are more informative while
higher frequencies are usually associated with  noise.

Let us consider an example. Suppose that   a data set is presented
by $10^{6}$ points. One way of data mining \cite{BN}, \cite{BMN}
is to convert the data set to a graph and develop harmonic
analysis associated with a corresponding combinatorial Laplace
operator. Let us assume for the simplicity that we identify our
data set with the path graph $\mathbb{Z}_{10^{6}}$ of $10^{6}$
vertices. We measure frequency on this graph in terms of the
eigenvalues of the corresponding Laplace operator on
$\mathbb{Z}_{10^{6}}$ whose definition is given in section 4. It
has $10^{6}$ eigenvalues (frequencies) which all belong to the
interval $[0, 4] $ and are given by the formula
$$
2-2\cos\frac{\pi k}{10^{6}-1},\>\>\>\>\> k=0,1,...,10^{6}-1.
$$
Our results show, that if one will delete every second point from
$\mathbb{Z}_{10^{6}}$ then the resulting set will be a uniqueness
set and even a sampling set (see definitions bellow) for all
functions on $\mathbb{Z}_{10^{6}}$ which are linear combinations
of  the (about) $12\times 10^{4}$  first eigenfunctions. If one
will delete about $2/3$ of all points then the resulting set is a
sampling set for all functions on $\mathbb{Z}_{10^{6}}$ which are
linear combinations of  the (about) $6\times 10^{4}$  first
eigenfunctions. By extending our reasoning it is possible to show
that about 10 percent of "uniformly distributed" points of
$\mathbb{Z}_{10^{6}}$ form a sampling set for functions on
$\mathbb{Z}_{10^{6}}$ which  are linear combinations of  the about
$500$  first eigenfunctions.

Thus by applying an appropriate filtering to a function  on a
graph, i.e. by removing high frequencies we not only remove noise
but we also reduce analysis on a whole graph to analysis on a much
smaller subgraph without loosing many of the lower frequencies. We
also give estimates of possible losses   of information which can
occur after filtering.

In order to construct a sampling theory on combinatorial graphs we
 prove certain analogs of Poincare inequality on graphs.
Our Poincare inequalities in the Section 2 provide estimates of
the norm of a function in terms of its "derivatives".

In what follows we introduce few basic notions and formulate and
discuss one of our Poincare inequalities.
 We consider finite
or infinite and in the later case countable connected graphs
$G=(V(G),E(G))$, where $V(G)$ is its set of vertices and $E(G)$ is
its set of edges. We consider only simple (no loops, no multiple
edges) undirected unweighted graphs. A number of vertices adjacent
to a vertex $v$ is called the degree of $v$ and denoted by $d(v)$.
We assume that all vertices have finite degrees but we do not
assume that the set of degrees of all vertices
$\left\{d(v)\right\}_{v\in V(G)}$ is bounded.

 The space $L_{2}(G)$ is the Hilbert space of all
complex-valued functions $f:V(G)\rightarrow \mathbb{C}$ with the
following inner product
$$
\left<f,g\right>= \sum_{v\in V(G)}f(v)\overline{g(v)}
$$
and the following norm
$$
\|f\|=\left(\sum_{v\in V(G)}|f(v)|^{2}\right)^{1/2}.
$$

By the adjacency matrix $A$ of $G$ we understand a matrix with
entries  $\{a_{uv}\}, u,v\in V(G)$, where $a_{uv}=1$, if vertices
$u$ and $v$ are adjacent, and $a_{uv}=0$ otherwise.

 Let $A$ be the adjacency matrix of $G$ and $D$
be a diagonal matrix whose entree on main diagonal are degrees of
the corresponding vertices. Then we consider the following version
of the discrete Laplace operator on $G$
\begin{equation}
L=D-A,
\end{equation}
or explicitly
$$
Lf(v)=\sum_{u\sim v}\left(f(v)-f(u)\right),f\in L_{2}(G),
$$
where notation $u\sim v$ means that $u$ and $v$ are adjacent
vertices. Note that this operator is different from the normalized
Laplace operator $\mathcal{L}$ is defined in \cite{Ch} and which
was considered in our previous papers  \cite{Pes3}-- \cite{Pes5}.

 The Laplace operator $L$ is self-adjoint and
positive definite in the space $L_{2}(G)$. Moreover, if degrees of
all vertices are uniformly bounded
\begin{equation}
D(G)=\max_{v\in V(G)}d(v)<\infty\label{degree condition}
\end{equation}
then the operator $L$ is bounded and its spectrum $\sigma(L)$ is a
subset of the interval $[0, 2D(G)]$. Note, that for the normalized
version of the Laplace operator $\mathcal{L}$ the  spectrum is
always a subset of $[0, 2]$. 

 Given a proper subset of vertices $S\subset V(G)$ its
vertex boundary $bS$ is the set of all vertices in $V(G)$ which
are not in $S$ but adjacent to a vertex in S
$$
b S=\left\{v\in V(G)\backslash S:\exists \{u,v\}\in E(G), u\in
S\right\}.
$$

If a graph $G=(V(G),E(G))$ is connected and $S$ is a proper subset
of $V$ then the vertex boundary $b S$ is not empty.

We will also use  the following notation
\begin{equation}
D(S)=\max_{v\in  S}d(v), S\subset V(G).
\end{equation}

To illustrate our Poincare inequalities let us formulate and
discuss a particular case of a more general inequality proved in
the Theorems \ref{MainPoincare} and \ref{PowerPoinc}. The
following inequality gives an estimate of the norm of a function
trough its "first order derivatives" and in this sense it can be
considered as a global Poincare inequality.
\begin{thm}
If $S$ is a subset of vertices such that every vertex $v$  in $bS$
is connected to at least $K_{0}=K_{0}(S)$ vertices in $S$ and
\begin{equation}
 \overline{S}=S\cup bS=V(G),
\end{equation}
then the following inequality holds for all $f\in L_{2}(G)$
\begin{equation}
\|f\|\leq \left\{\sum_{u\in
S}\left(\frac{2d_{0}(u)}{K_{0}}+1\right)|f(u)|^{2}\right\}^{1/2}+
\frac{2}{\sqrt{K_{0}}}\left\|L^{1/2} f\right\|,\label{1}
\end{equation}
where $d_{0}(u),\>\>u\in S,$ is the number of vertices in $bS$
adjacent to $u\in S$. \label{T.1.1}
\end{thm}
\begin{exmp}

Suppose that $G$ is a star $\{v_{0}, v_{1}, ..., v_{N}\}$ whose
center is $v_{0}$. Let $S$ be the vertex $\{v_{0}\}$. Then
$K_{0}=1, d_{0}(v_{0})=N$, and (\ref{1}) becomes
$$
\|f\|\leq \sqrt{2N+1}|f(v_{0})|+2\|L^{1/2}f\|.
$$
In particular, if  $f(v_{j})=1$ for all $ 0\leq j\leq N$ then $
\|f\|=\sqrt{N+1}, \|L^{1/2} f\|=0$, and (\ref{1}) becomes
$$
\sqrt{N+1}\leq \sqrt{2N+1}.
$$
\end{exmp}

\begin{exmp}
For the same star graph as above we consider
$S=\{v_{1},...,v_{N}\}$, then $K_{0}(S)=N, d_{0}(v_{0})=1, $  and
for any $N$  the inequality (\ref{1}) becomes
$$
\|f\|\leq \sqrt{\frac{2}{N}+1}\sqrt{\sum_{u\in
S}|f(u)|^{2}}+\frac{2}{\sqrt{N}}\|L^{1/2}f\|.
$$
In particular, if  we consider function $f$ such that $f(v_{0})=0$
and $f(v_{j})=1$  for all other $j=1,2,..,N$, then $
\|f\|=\sqrt{N}, \|L^{1/2} f\|=\sqrt{N}$, and for any $N$  the
inequality (\ref{1}) becomes
$$
\sqrt{N}\leq \sqrt{N+2}+2.
$$
For the function $f$ which is identical one, we obtain
$$
\sqrt{N+1}\leq \sqrt{N+2}.
$$
\end{exmp}

\begin{exmp}
Let $C_{N}$ be a cycle of $N$ vertices $\{v_{1},...,v_{N}\}$. Take
another vertex $v_{0}$ and make a graph $C_{N}\cup \{v_{0}\}$ by
connecting $v_{0}$ to each of $v_{1},...,v_{N}$.

 Let
$\lambda_{k}(N)$ be a
  non-zero eigenvalue of the operator $L$ on the graph $C_{N}$
and let $\varphi_{k}$ be a corresponding orthonormal
eigenfunction. Construct a function $\widetilde{\varphi_{k}}$ on
the graph $C_{N}\cup \{v_{0}\}$  such that
$\widetilde{\varphi}_{k}(v)=\varphi_{k}(v)$ if $v\in C_{N}$ and
$\widetilde{\varphi}_{k}(v_{0})=0$. Since $\varphi_{k}$ is
orthogonal to the constant function $\textbf{1}$ we have that
$$
\sum_{v_{j}\in C_{N}}\varphi_{k}(v_{j})=0
$$
and it implies that for the operator $L$ on $C_{N}\cup \{v_{0}\}$
$$
L\widetilde{\varphi}_{k}(v_{0})=0.
$$
Clearly,  for every $v_{j}\in C_{N}$ one has
$$
L\widetilde{\varphi}(v_{j})=L\varphi_{k}(v_{j})+\varphi(v_{j})=
\left(\lambda_{k}(N)+1)\varphi(v_{j}\right) .
$$ Thus,
$$
L\widetilde{\varphi}_{k}=\left(\lambda_{k}(N)+1\right)\widetilde{\varphi}_{k},
$$
and since $\|\widetilde{\varphi}_{k}\|=1$ we have that
$$
\|L^{1/2}\widetilde{\varphi}_{k}\|=\left(\lambda_{k}(N)+1\right)^{1/2}.
$$
Let $S$ be the graph $C_{N}=\{v_{1},...,v_{N}\}$. In this case the
boundary of $S$ is the point $v_{0}$,  $K_{0}=N$, $d_{0}(v_{j})=1$
and then for the function $\widetilde{\varphi}_{k}$ the inequality
(\ref{1}) takes the following form
$$
1\leq \sqrt{\frac{2}{N}+1}+2\sqrt{\frac{\lambda_{k}(N)+1}{N}}.
$$
Since for $k=1$ the eigenvalue $\lambda_{1}(N)$ goes to zero when
$N$ goes to infinity, we see that the right-hand side of the last
inequality can be made arbitrary close to one.

For any $k=1,...,N$  one has the estimate $\lambda_{k}(N)\leq 2N$
and for the corresponding $\widetilde{\varphi}_{k}$ it gives the
inequality
$$
1\leq \sqrt{\frac{2}{N}+1}+2\sqrt{\frac{1}{N}+2}.
$$
\end{exmp}

According to our definition of Paley-Wiener functions (see Section
3, and also \cite{Pes7}-\cite{Pes3}) they always satisfy the
Bernstein inequality and together with Poincare inequality it
leads to Plancherel-Polya inequalities on graphs.

Our Poincare-Polya-type inequalities (Theorem \ref{GraphPP}) give
two-sided estimate of the norm of an appropriate Paley-Wiener
function in terms of its values on a subset of vertices. We use
these estimates to apply classical ideas of Duffin and Schaeffer
\cite{DS} about Hilbert frames to obtain the sampling Theorem
\ref{SamplingTheorem} which is one of the main results of the
paper. In particular we obtain a formula which represents
Paley-Wiener functions in terms of their values on specific
subgraphs. We call them  \textbf{sparse representations } of
Paley-Wiener function.

In Section 4 we construct a filtering operator (Theorem \ref{FT})
using Schrodinger's one-parameter group of operators generated by
a self-adjoint positive definite operator $L$ in the Hilbert space
$L_{2}(G)$. This filtering operator maps entire Hilbert space
$L_{2}(G)$ into appropriate Paley-Wiener space. We also prove our
version of the Direct Approximation Theorem using a modulus of
continuity expressed in terms of the Schrodinger group of
operators generated by $L$ (Theorem \ref{DAT}). By combining
filtering procedure with our  sampling theory, we obtain
\textbf{sparse approximations} to functions in $L_{2}(G)$.

 We would like to emphasize that the
sampling theory that is developed in the present article is
different from the one we had developed in \cite{Pes3},
\cite{Pes4}, \cite{Pes5}, \cite{PesPes}. We also have to mention
that our approach to sampling on graphs is very different from
methods which were presented and explored in \cite{FrT}, \cite{G},
\cite{MSW}. Note, that our approximation theory on graphs is a
generalization of the classical approximation theory by
Paley-Wiener functions \cite{Akh}, \cite{N}. It also has to be
mentioned that  some ideas about approximation theory on compact
metric spaces (which include finite graphs) were introduced in
\cite{MM}.  Basic ideas of harmonic analysis on graphs that are
relevant to our paper were recently summarized in the book
\cite{Mah}.

Our results can have  applications to filtering, denoising,
approximation and compression of functions on graphs. These tasks
are of central importance to  data dimension reduction, image
processing, computer graphics, visualization and learning theory.

\bigskip

\section{ Poincare inequalities on combinatorial graphs}

For a function $f\in L_{2}(G)$ we introduce a measure of
smoothness which is the  norm of a "gradient"
\begin{equation}
\left\|\nabla f\right\|^{2}=\sum_{v\sim u}|f(v)-f(u)|^{2},
\end{equation}
where the sum is taken over all unordered pairs$\{v, u\}$ for
which $v$ and $u$ are adjacent.  Given a subset $W\subset V$ we
will use the notation
\begin{equation}
\left\|\nabla f\right\|^{2}_{W}=\sum_{v\sim u, v,u\in
W}|f(v)-f(u)|^{2}.
\end{equation}
 For any $S$ which is a subset of vertices of $G$
we introduce the following operator
\begin{equation}
cl^{0}(S)=S,\>\>cl(S)=S\cup bS,
cl^{m}(S)=cl\left(cl^{m-1}(S)\right), m\in \mathbb{N}, S\subset
V(G).
\end{equation}
We will use the following notion of the relative degree.
\begin{defn}
 For a vertex $v\in cl^{m}(S)$ we introduce the
relative degree $d_{m}(v)$ as the number of vertices in the
boundary $b\left(cl^{m}(S)\right)$ which are adjacent to $v$:
$$
d_{m}(v)=card\left\{w\in b\left(cl^{m}(S)\right): w\sim v\right\}.
$$
 For any $ S\subset V(G)$ we introduce the following
notation
$$
 D_{m}=D_{m}(S)=\sup_{v\in cl^{m} (S)}d_{m}(v).
$$
\end{defn}

\begin{defn}
 For a vertex $v\in b\left(cl^{m}(S)\right)$ we introduce the
quantity  $k_{m}(v)$ as the number of vertices in the set
$cl^{m}(S)$ which are adjacent to $v$:
$$
k_{m}(v)=card\left\{w\in cl^{m}(S): w\sim v\right\}.
$$
For any $ S\subset V(G)$ we introduce the following notation
$$
 K_{m}=K_{m}(S)=\inf_{v\in      b\left(cl^{m}(S)\right) }k_{m}(v).
$$

\end{defn}
For a given set $S\subset V(G)$ and a fixed $n\in \mathbb{N}$
consider a sequence of closures
$$
S, cl(S),..., cl^{n}(S), n\in \mathbb{N}.
$$

\begin{thm}
In the same notations as above, if $S$ is a subset of vertices
such that the boundary of $cl^{n-1}(S), n\in \mathbb{N},$ is not
empty then the following inequality holds
\begin{equation}
\left(\sum_{v\in cl^{n}(S)}|f(v)|^{2}\right)^{1/2}\leq
\left(\prod_{i=0}^{n-1}\left(\frac{2D_{i}}{K_{i}}+1\right)\right)^{1/2}\left(\sum_{v\in
S}|f(v)|^{2}\right)^{1/2}+
$$
$$
2\left(\sum_{j=0}^{n-1}\frac{1}{K_{j}}\prod_{i=j+1}^{n-1}\left(\frac{2D_{i}}{K_{i}}+1\right)
\right)^{1/2}\|L^{1/2}f\|.\label{nabla10}
\end{equation}
In particular, if
\begin{equation}
cl^{n}(S)=V(G),\label{assumption}
\end{equation}
then

\begin{equation}
\|f\|\leq
\left(\prod_{i=0}^{n-1}\left(\frac{2D_{i}}{K_{i}}+1\right)\right)^{1/2}\left(\sum_{v\in
S}|f(v)|^{2}\right)^{1/2}+
$$
$$
2\left(\sum_{j=0}^{n-1}\frac{1}{K_{j}}\prod_{i=j+1}^{n-1}\left(\frac{2D_{i}}{K_{i}}+1\right)
\right)^{1/2}\|L^{1/2}f\|.\label{MPI}
\end{equation}
\label{MainPoincare}
\end{thm}
\begin{proof}

First, we are going to prove, that for any subset $S$ of vertices
for which the boundary of $cl^{n-1}(S), n\in \mathbb{N},$ is not
empty the following inequality holds

\begin{equation}
\sum_{w_{n}\in cl^{n}(S)}|f(w_{n})|^{2}\leq
\prod_{i=0}^{n-1}\left(\frac{2D_{i}}{K_{i}}+1\right)\sum_{w_{0}\in
S}|f(w_{0})|^{2}+
$$
$$
\frac{2}{K_{0}}\prod_{i=1}^{n-1}\left(\frac{2D_{i}}{K_{i}}+1\right)\sum_{v_{0}\in
bS}\sum_{j_{0}=1}^{K_{0}(S)}|f(v_{0})-f(u_{j_{0}}(v_{0}))|^{2}+
$$
$$
\frac{2}{K_{1}}\prod_{i=2}^{n-1}\left(\frac{2D_{i}}{K_{i}}+1\right)\sum_{v_{1}\in
b(cl(S))}\sum_{j_{1}=1}^{K_{1}}|f(v_{1})-f(u_{j_{1}}(v_{1}))|^{2}+...
$$
$$
\frac{2}{K_{n-1}}\sum_{v_{n-1}\in
b(cl^{n-1})}\sum_{j_{n-1}=1}^{K_{n-1}}|f(v_{n-1})-f(u_{j_{n-1}}(v_{n-1}))|^{2},\label{n-th
ineq1}
\end{equation}
where for every $0\leq m\leq n-1$ the  $u_{j_{m}}(v_{m})$ is a
vertex in $cl^{m}(S)$ which is adjacent to $v_{m}\in
b(cl^{m}(S))$.\label{fundlemma}

 For any two vertices $v,u\in V$ one has
\begin{equation}
f(v)=f(u)+\left(f(v)-f(u)\right)
\end{equation}
and
\begin{equation}
|f(v)|^{2}\leq 2\left(|f(u)|^{2}+|f(v)-f(u)|^{2}\right).
\end{equation}
Let $v\in b S$ and $u_{1}(v),..., u_{K_{0}}(v)\in S$ be some of
the $K_{0}$ vertices in $S$ which are  adjacent to $v$. For each
of them the following inequality holds
\begin{equation}
|f(v)|^{2}\leq
2\left(|f(u_{j}(v))|^{2}+|f(v)-f(u_{j}(v))|^{2}\right), 1\leq j
\leq K_{0},
\end{equation}
which implies the inequality
\begin{equation}
|f(v)|^{2}\leq
\frac{2}{K_{0}}\sum_{j=1}^{K_{0}}|f(u_{j}(v))|^{2}+\frac{2}{K_{0}}\sum_{j=1}^{K_{0}}|f(v)-f(u_{j}(v))|^{2}.
\end{equation}
 Since  every $u_{j}(v)\in S, 1\leq
j \leq K_{0},$ can be adjacent to a maximum of
  $d_{0}(u_{j}(v))$ distinct vertices $v\in bS$,  the previous
 inequality implies
\begin{equation}
\sum_{v\in bS}|f(v)|^{2}\leq
$$
$$
\sum_{j=1}^{K_{0}}\frac{2}{K_{0}}\sum_{v\in b
S}|f(u_{j}(v))|^{2}+\sum_{j=1}^{K_{0}}\frac{2}{K_{0}}\sum_{v\in b
S}|f(v)-f(u_{j}(v))|^{2}\leq
$$
$$
\sum_{u\in
S}\frac{2d_{0}(u)}{K_{0}}|f(u)|^{2}+\sum_{j=1}^{K_{0}}\frac{2}{K_{0}}\sum_{v\in
b S}|f(v)-f(u_{j}(v))|^{2}.\label{intermediate}
\end{equation}
Thus,
 \begin{equation}
  \sum_{v\in bS}|f(v)|^{2}\leq \sum_{u\in
S}\frac{2d_{0}(u)}{K_{0}}|f(u)|^{2}+\sum_{j=1}^{K_{0}}\frac{2}{K_{0}}\sum_{v\in
b S}|f(v)-f(u_{j}(v))|^{2}, \label{betterineq}
\end{equation}
where $u_{1}(v),..., u_{K_{0}}$ are  different vertices from $S$
that adjacent to $v$. The last inequality implies the following

\begin{equation}
  \sum_{v\in bS}|f(v)|^{2}\leq \sum_{u\in
S}\frac{2D_{0}}{K_{0}}|f(u)|^{2}+\sum_{j=1}^{K_{0}}\frac{2}{K_{0}}\sum_{v\in
b S}|f(v)-f(u_{j}(v))|^{2}.
\end{equation}
By adding this inequality with the identity
$$
\sum_{v\in  S}|f(v)|^{2}=\sum_{v\in  S}|f(v)|^{2},
$$
 one obtains the inequality
 which holds true for any subset of vertices $S$:
\begin{equation}
\sum_{w_{1}\in cl(S)}|f(w_{1})|^{2}\leq
\left(\frac{2D_{0}}{K_{0}}+1\right)\sum_{w_{0}\in
S}|f(w_{0})|^{2}+
$$
$$\frac{2}{K_{0}}\sum_{v_{0}\in bS}\sum_{j=1}^{K_{0}}|f(v_{0})-f(u_{j}(v_{0}))|^{2}
,\label{1-st ineq}
\end{equation}
where $u_{1}(v_{0}),..., u_{K_{0}}\in  S$ are different and
adjacent to $v$.

 Since $cl^{2}(S)=cl(cl(S))$ the inequality (\ref{1-st ineq}) implies the
following one
\begin{equation}
\sum_{w_{2}\in cl^{2}(S)}|f(w_{2})|^{2}\leq
\left(\frac{2D_{1}}{K_{1}}+1\right) \sum_{w_{1}\in
cl(S)}|f(w_{1})|^{2}+
$$
$$
\frac{2}{K_{1}}\sum_{v_{1}\in
b(cl(S))}\sum_{j_{1}=1}^{K_{1}}|f(v_{1})-f(u_{j_{1}}(v_{1}))|^{2},
\end{equation}
where $u_{j_{1}}(v_{1})\in cl(S)$. Along with the (\ref{1-st
ineq}) it gives
\begin{equation}
\sum_{w_{2}\in cl^{2}(S)}|f(w_{2})|^{2}\leq
\left(\frac{2D_{1}}{K_{1}}+1\right)\left(\frac{2D_{0}}{K_{0}}+1\right)\sum_{w_{0}\in
S}|f(w_{0})|^{2}+
$$
$$
\frac{2}{K_{0}}\left(\frac{2D_{1}}{K_{1}}+1\right)\sum_{v_{0}\in
bS}\sum_{j_{0}=1}^{K_{0}}|f(v_{0})-f(u_{j_{0}}(v_{0}))|^{2}+
$$
 $$
\frac{2}{K_{1}}\sum_{v_{1}\in
b(cl(S))}\sum_{j_{1}=1}^{K_{1}}|f(v_{1})-f(u_{j_{1}}(v_{1}))|^{2},\label{2-nd
ineq}
\end{equation}
where $u_{j_{1}}(v_{1})\in cl(S), u_{j_{0}}(v_{0})\in S$ are
different vertices that adjacent to $v$.

The derivation of (\ref{2-nd ineq}) shows that by induction one
can prove the inequality (\ref{n-th ineq1}).

Next, let us remind, that  just by construction the vertices
$v_{m}\in b\left(cl^{m}(S)\right)$ and $u_{j_{m}}(v_{m})\in
cl^{m}(S)$ are adjacent and
 $$
u_{k_{1}}(v_{m})\neq u_{k_{2}}(v_{m}),
$$
if $k_{1}\neq k_{2}$. It is also clear that $v_{m}\in
b\left(cl^{m}(S)\right)$ is different from any of $v_{k}\in
b\left(cl^{k}(S)\right)$  as long as $m\neq k$.  Because if this
the inequality (\ref{n-th ineq1}) implies the inequality
\begin{equation}
\|f\|\leq
\left(\prod_{i=0}^{n-1}\left(\frac{2D_{i}}{K_{i}}+1\right)\right)^{1/2}\left(\sum_{v\in
S}|f(v)|^{2}\right)^{1/2}+
$$
$$
\left(\sum_{j=0}^{n-1}\frac{2}{K_{j}}\prod_{i=j+1}^{n-1}\left(\frac{2D_{i}}{K_{i}}+1\right)
\right)^{1/2}\|\nabla f\|.\label{nabla inequality}
\end{equation}

To prove the Theorem we are going to show that
\begin{equation}
\|\nabla f\|=\sqrt{2}\|L^{1/2}f\|.\label{fund}
\end{equation}
Indeed,

\begin{equation}
 \left\|\nabla
f\right\|^{2}=\sum_{v\sim u}|f(v)-f(u)|^{2}=
$$
$$
\sum_{v\in V}|f(v)|^{2}d(v)+
 \sum_{u\in V}|f(u)|^{2}d(u)-
 2\sum_{v\sim u}f(v)f(u)=
 $$
 $$
2\left(
  \sum_{v\in V}|f(v)|^{2}d(v)-
   \sum_{v\sim u}f(v)f(u)
   \right)=
   $$
   $$2\left(\left<Df,f\right>-
  \left<Af,f\right>\right)=
 2\left<L
f,f\right>= 2\left<L^{1/2}
f,L^{1/2}f\right>=2\|L^{1/2}f\|^{2}\label{gradient}.
\end{equation}
The inequalities (\ref{nabla inequality}) and (\ref{gradient})
imply (\ref{nabla10}). The Theorem is proved.

\end{proof}
The formula (\ref{betterineq}) gives the following Corollary.
\begin{col}
If $S$ is a subset of vertices such that every vertex $v$  in $bS$
is connected to at least $K_{0}(S)$ vertices in $S$, where $1\leq
K_{0}(S)\leq d(v)$ then for all $f\in L_{2}(G)$
\begin{equation}
\left\{\sum_{v\in cl(S)}|f(v)|^{2}\right\}^{1/2}\leq \left\{
\sum_{u\in
S}\left(\frac{2d_{0}(u)}{K_{0}(S)}+1\right)|f(u)|^{2}\right\}^{1/2}+
\frac{2}{\sqrt{K_{0}(S)}}\left\|L^{1/2} f\right\|.\label{poinc}
\end{equation}

If in addition
\begin{equation}
 cl(S)=S\cup bS=V(G),\label{cl1}
\end{equation}
then the following inequality holds for all $f\in L_{2}(G)$
\begin{equation}
\|f\|\leq \left\{\sum_{u\in
S}\left(\frac{2d_{0}(u)}{K_{0}(S)}+1\right)|f(u)|^{2}\right\}^{1/2}+
       \frac{2}{\sqrt{K_{0}(S)}}         \left\|L^{1/2}
       f\right\|.\label{cl2}
\end{equation}

\end{col}

Since $K_{0}(S)\geq 1$ we have the following.

\begin{col}
 If $S$ is a subset of
vertices such that the condition (\ref{cl1})  holds then
 the following inequality takes place
\begin{equation}
\|f\|\leq \left \{\sum_{u\in
S}\left(2d_{0}(u)+1\right)|f(u)|^{2}\right\}^{1/2} +2\|L^{1/2}f\|.
\end{equation}
\end{col}

To extend our results to higher powers of $L$ we will need the
following Lemma \cite{Pes7}--\cite{Pes4}.

\begin{lem}
 If for some  positive $c>0, a>0, s>0,$ and an   $\varphi \in
L_{2}(G)$  the following inequality holds true
\begin{equation}
\|\varphi\|\leq a+c\|L^{s}\varphi\|,\label{lemma condition 2}
 \end{equation}
 then for the same $c,
a, s, \varphi$ the following holds
\begin{equation} \|\varphi\|\leq
2ra+8^{r-1}c^{r}\|L^{rs}\varphi\|\label{induction}
\end{equation}
for all $r=2^{l}, l=0, 1, ... $ .\label{exp}
\end{lem}
An application of the Lemma \ref{exp} gives the following result.
\begin{thm}
If the assumption (\ref{assumption}) is satisfied, then for any
$r=2^{l}, l=0,1,...,$ the next inequality holds
\begin{equation}
\|f\|\leq 2r\left(\prod_{i=0}^{n-1}\left(\frac{2D_{i}}{K_{i}}+
1\right)\right)^{1/2}\left(\sum_{v\in S}|f(v)|^{2}\right)^{1/2}+
$$
$$
2^{4r-3}\left(\sum_{j=0}^{n-1}\frac{1}{K_{j}}\prod_{i=j+1}^{n-1}\left(\frac{2D_{i}}{K_{i}}+1\right)
\right)^{r/2}\|L^{r/2}f\|.
\end{equation}\label{PowerPoinc}
\end{thm}

\begin{thm}
 If $S$ is a subset of
vertices such that every vertex $v$  in $bS$ is connected to at
least $K_{0}$ vertices in $S$, where $1\leq K_{0}\leq d(v)$ and
the condition (\ref{cl1}) holds, then for any $r=2^{l},
l=0,1,...,$
\begin{equation}
\|f\|\leq 2r\sqrt{\frac{2D_{0}}{K_{0}}+1}\left(\sum_{u\in
S}|f(u)|^{2}\right)^{1/2}+\frac{2^{4r-3}}{K_{0}^{r/2}}
\|L^{r/2}f\|\label{4.2}.
\end{equation}

\end{thm}
\begin{proof} According to the previous Theorem the assumptions of
the Theorem give the inequality
\begin{equation}
\|f\|\leq \sqrt{\frac{2D_{0}}{K_{0}}+1}\left(\sum_{u\in
S}|f(u)|^{2}\right)^{1/2}+\frac{2}{\sqrt{K_{0}}}\|L^{1/2}f\|\label{4.2}.
\end{equation}
Now an application of (\ref{induction}) gives the result.

\end{proof}

\begin{col}
If $S$ is a subset of vertices such that (\ref{cl1}) holds, then
for any  $r=2^{l}, l=0,1,...,$ the following holds
\begin{equation}
\|f\|\leq 2r\sqrt{2D_{0}+1}\left(\sum_{u\in
S}|f(u)|^{2}\right)^{1/2}+2^{4r-3}\|L^{r/2}f\|\label{4.2}.
\end{equation}
\end{col}

\bigskip

\section{Plancherel-Polya inequalities,  sampling and sparse representations
of Paley-Wiener functions}

By the spectral theory of self-adjoint operators \cite{BS}, there
exist a direct integral of Hilbert spaces $A=\int A(\lambda )dm
(\lambda )$ and a unitary operator $\mathcal{F}_{L}$ from
$L_{2}(G)$ onto $A$, which transforms the domain $\mathcal{D}_{s},
s\geq 0,$ of the operator $ L^{s}$ onto $A_{s}=\{a\in
A|\lambda^{s}a\in A \}$ with norm
\begin{equation}
\|a(\lambda )\|_{A_{s}}= \left (\int^{\infty}_{0} \lambda ^{2s}
\|a(\lambda )\|^{2}_{A(\lambda )} dm(\lambda ) \right
)^{1/2}\label{SpecFT}
 \end{equation}
 and
$\mathcal{F}_{L}(Lf)=\lambda (\mathcal{F}_{L}f), f\in
\mathcal{D}_{1}. $
\begin{defn}
The unitary operator $\mathcal{F}_{L}$ will be called the Spectral
Fourier transform and $a=\mathcal{F}_{L}f $ will be called the
Spectral Fourier transform of $f\in L_{2}(G)$. \label{Def2}
\end{defn}
\begin{defn}
We will say that a function $f$ in  $L_{2}(G)$  belongs to the
space $PW_{\omega}(L)$  if its Spectral Fourier transform
$\mathcal{F}_{L}f=a$ has support in $[0, \omega ] $.
\label{DefSFT}
\end{defn}
The following  theorem describes some basic properties of
Paley-Wiener vectors and show that they share similar properties
to those of the classical Paley-Wiener functions. The proof of
these and many other properties of Paley-Wiener vectors  can be
found in our other papers and in particular in \cite{Pes7}-
\cite{Pes4}.
\begin{thm}
The following conditions are equivalent:

1) The linear set $\bigcup _{ \omega >0}PW_{\omega }(L)$ is dense
in $L_{2}(G)$.

2) The set $PW_{\omega }(L)$ is a linear closed subspace in
$L_{2}(G)$.

3)  A function $f$ belongs to a space $ PW_{\omega}(L)$ if and
only if  for all $k\in \mathbb{N},$ the following Bernstein
inequality holds true
\begin{equation}
\|L^{k}f\|\leq \omega^{k}\|f\|;\label{BI}
\end{equation}
\end{thm}

To obtain a Sampling Theorem for Paley-Wiener functions on graphs
we have to establish  Plancherel-Polya-type inequalities.
 The
inequalities  (\ref{MPI}) and (\ref{BI}) along with the obvious
inequality
\begin{equation}
\left(\sum_{u\in S}|f(u)|^{2}\right)^{1/2}\leq \|f\|
\end{equation}
 imply the following Plancherel-Polya-type  inequalities for
 functions in $PW_{\omega}(L)$.
\begin{thm}In the same notations as in the Theorem \ref{MainPoincare},
if the condition
\begin{equation}
cl^{n}(S)=V(G)
\end{equation}
and the inequality
\begin{equation}
\omega<\frac{1}{4}\left(\sum_{j=0}^{n-1}\frac{1}{K_{j}}\prod_{i=j+1}^{n-1}\left(\frac{2D_{i}}{K_{i}}+1\right)
\right)^{-1}
\end{equation}
hold, then for any $f\in PW_{\omega}(L)$ the next inequality takes
place
\begin{equation}
\left(\sum_{u\in S}|f(u)|^{2}\right)^{1/2}\leq \|f\|\leq
\frac{1}{1-\gamma}\left(\prod_{i=0}^{n-1}\left(\frac{2D_{i}}{K_{i}}+1\right)\right)^{1/2}
\left(\sum_{u\in
S}|f(u)|^{2}\right)^{1/2},
\end{equation}
where
\begin{equation}
\gamma=2\omega^{1/2}
\left(\sum_{j=0}^{n-1}\frac{1}{K_{j}}\prod_{i=j+1}^{n-1}\left(\frac{2D_{i}}{K_{i}}+1\right)
\right)^{1/2}<1.
\end{equation}
\end{thm}

In particular we have the following Theorem.
\begin{thm}
 If $S$ is a subset of
vertices such that every vertex $v$  in $bS$ is adjacent to at
least $K_{0}$ vertices in $S$, every $v\in S$ is adjacent to at
most $D_{0}$ vertices in $bS$ and the condition
\begin{equation}
\overline{S}=S\cup bS=V(G),\label{cl}
\end{equation}
along with the inequality
\begin{equation}
\omega<\frac{K_{0}}{4}\label{w-ineq}
\end{equation}
hold, then
\begin{equation}
\left(\sum_{u\in S}|f(u)|^{2}\right)^{1/2}\leq \|f\|\leq
\frac{1}{1-\gamma}\sqrt{\frac{2D_{0}}{K_{0}}+1}\left(\sum_{u\in
S}|f(u)|^{2}\right)^{1/2},\label{PP-2}
\end{equation}
where $f\in PW_{\omega}(L)$
 and $\gamma=2\sqrt{\omega/K_{0}}<1$.\label{GraphPP}
\end{thm}

The significance of the
 inequalities (\ref{PP-2}) is that they give two-sided estimate of
 the norm $\|f\|$ of a Paley-Wiener function $f$ in terms of its
 values on a smaller set $S\subset V(G)$.

Let $P_{\omega}$ be the orthogonal projector
 $$
P_{\omega}: L_{2}(G)\rightarrow PW_{\omega}(L).
$$
 The last inequality (\ref{PP-2}) shows that the set of functions
$P_{\omega}(\delta_{u}), u\in S$, where $\delta_{u}$ is the Dirac
measure concentrated at $u\in S$, is a Hilbert frame in the
Hilbert space $PW_{\omega}(L)$ when $ \omega<\frac{K_{0}}{4}. $

Thus, by applying the classical result of Duffin and Schaeffer
\cite{DS} about dual frames we obtain the following uniqueness and
reconstruction Theorem. For the sake of simplicity we formulate it
just for particular situation that satisfies (\ref{cl}).

\begin{thm}
 If $S$ is a subset of
vertices such that every vertex $v$  in $bS$ is connected to at
least $K_{0}$ vertices in $S$, where $1\leq K_{0}\leq d(v)$ and
the condition
\begin{equation}
\overline{S}=S\cup bS=V(G),
\end{equation}
 along with
$$
\omega<\frac{K_{0}}{4}
$$
hold, then

1) the set $S$ is a uniqueness set for functions in
$PW_{\omega}(L)$;

 2) there exist functions $\Theta_{u}\in PW_{\omega}(L)$, $u\in S,$
such that for all $f\in PW_{\omega}(L)$ the following
reconstruction formula holds
\begin{equation}
f=\sum_{u\in S}f(u)\Theta_{u}.\label{sparserepresentation}
\end{equation}
\label{SamplingTheorem}
\end{thm}

The last formula (\ref{sparserepresentation}) is what we call a
\textbf{sparse representation} since it represents a function
through its values on a subgraph.
\begin{rem}
It is clear that if the spectrum the Laplace operator of a graph $G$  is very close to zero  then there are many subsets $S$ of $G$ and  functions on $G$  for which the last two Theorems convey non-trivial information. Thus, in the section bellow we discuss the situation on the infinite graph $\mathbb{Z}_{n}$.

However, in the case of a finite graph $G $ the spectral resolution is just the eigenvalue-eigenfunction representation and if $K_{0}/4$ is less than the first strictly positive eigenvalue of $L$, then the assumption $\omega<K_{0}/4$ would satisfy only for constant functions on $G$ and the above  inequalities would be trivial.

But there are many finite graphs for which the Theorems 3.3 and 3.4 are not trivial for a "right" choice of subsets $S\subset V(G)$.
For example, take the cycle $C_{n}$ of $n$ vertices for which the eigenvalues of the corresponding Laplace operator are $2-2\cos \frac{2\pi k}{n}$. For a large $n$ there many eigenvalues which are very close to zero. On the other hand there are  sets $S$ of "isolated" points in $C_{n}$ for which the number $K_{0}$ is either $1$ or $2$. Thus, the previous Theorems hold true for any of such  sets of points for functions from $PW_{\omega}(G)$ , for which either $\omega<1/4$ or $\omega<2/4$.

\end{rem}

\section{Filtering and   Direct Approximation Theorem
by Paley-Wiener functions on graphs}

The goal of the section is to describe relations between
Schr\"{o}dinger's Semigroup $e^{itL}$ and the functional
$$
E(f,\omega)=\inf_{g\in PW_{\omega}(L)}\|f-g\|,
$$
which measures a best approximation of $f\in L_{2}(G)$ by
functions from the Paley-Wiener space $PW_{\omega}(L), \omega\geq
0.$ If $f_{\omega}$ is the orthogonal projection of $f$ on
$PW_{\omega}(L)$, then according to (\ref{SpecFT}) one has the
following relation
\begin{equation}
 E(f,\omega)=\|f-f_{\omega}\|=\left
(\int_{\omega}^{\infty}  \|x(\tau )\|^{2}_{X(\tau)} d m(\tau)
\right )^{1/2},
\end{equation}
where $x(\tau )=\mathcal{F}f$ is the Spectral Fourier transform of
$f$.  In other words the best approximation $E(f,\omega)$ shows
the "rate of decay" of the Spectral Fourier transform
$\mathcal{F}f$ of $f$. The same formula (\ref{SpecFT})  implies
the following  inequality
\begin{equation}
 E(f,\omega)=\|f-f_{\omega}\|=\left
(\int_{\omega}^{\infty}  \|x(\tau )\|^{2}_{X(\tau)} d m(\tau)
\right )^{1/2}\leq
$$
$$
\omega^{-s} \left (\int_{0}^{\infty} \tau^{2s} \|x(\tau
)\|^{2}_{X(\tau)} d m(\tau) \right )^{1/2}=\omega^{-s}\|L^{s}f\|,
s>0,\label{ApprEst}
\end{equation}
for all  functions in   $L_{2}(G)$.

The quantity $\|L^{s}f\|, s>0, $ is  a measure of smoothness of a
function $f$. In this sense  the estimate (\ref{ApprEst})
generalizes the well-known fact of the classical harmonic analysis
that the rate of approximation of a function by Paley-Wiener
functions depends on the smoothness of this function.

 For any $f\in L_{2}(G)$ we
introduce a difference operator of order $m\in \mathbb{N}$ as
$$
\Delta^{m}_{s}f=\sum^{m}_{j=0}(-1)^{m+j}C^{j}_{m}e^{ijsL}f,
$$
where $C^{j}_{m}$ is the number of combinations from $m$ elements
taking $j$ at a time.  The modulus of continuity is defined as
$$
\Omega_{m}(f,s)=\sup_{|\tau|\leq
s}\left\|\Delta^{m}_{\tau}f\right\|.
$$

In the following Theorem we construct a filtering operator which
maps $H$ into a Paley-Wiener space.

\begin{thm}If $h\in L_{1}(\mathbb{R})$ is an entire function of exponential
type $\omega$ then for any $f\in  L_{2}(G)$ the function
$$
\mathcal{Q}_{h}^{\omega}(f)=\int _{-\infty}^{\infty}h(t)e^{itL}fdt
$$
belongs to $PW_{\omega}(L).$ \label{FT}
\end{thm}
\begin{proof}

If $g= \mathcal{Q}_{h}^{\omega}(f)$ then for every real $\tau$ we
have
$$
e^{i\tau
L}g=\int_{-\infty}^{\infty}h(t)e^{i(t+\tau)L}fdt=\int_{-\infty}^{\infty}h(
t-\tau)e^{itL}fdt.
$$
 Using this formula we can extend the abstract function $e^{i\tau L}g$ to the
complex plane as
$$
e^{izL}g=\int_{-\infty}^{\infty}h(t-z)e^{itL}fdt.
$$
 Since by assumption $h\in L_{1}(\mathbb{R})$ is an entire function of exponential
type $\omega$
 we have
 $$
 h(x+iy)=\sum_{0}^{\infty}\frac{(iy)^{k}}{k!}h^{(k)}(x)
  $$
  and the $L_{1}(\mathbb{R})$-Bernstein inequality implies the following 
  $$
  \int_{-\infty}^{\infty}|h(t-z)|dt\leq e^{\omega|z|}
 \int_{-\infty}^{\infty}|h(t)|dt.
  $$
  Thus, we obtain the following inequality 
$$
\|e^{izL}g\|\leq
\|f\|\int_{-\infty}^{\infty}|h(t-z)|dt\leq\|f\|e^{\omega|z|}
 \int_{-\infty}^{\infty}|h(t)|dt.
 $$
It shows that for every function $g^{*}\in L_{2}(G)$ the function
$\left<e^{izL}g,g^{*}\right>$ is an entire function and
$$
\left|\left<e^{izL}g,g^{*}\right>\right|\leq
\|g^{*}\|\|f\|e^{\omega|z|}\int_{-\infty}^{\infty}|h(t)|dt.
$$
In other words the $\left<e^{izL}g,g^{*}\right>$ is an entire
function of the exponential type $\omega$
  which is bounded on the real line and an application of the classical
Bernstein theorem gives the inequality
$$
\left|\left(\frac{d}{dt}\right)^{k}\left<e^{itL}g,g^{*}\right>\right|
\leq\omega^{k}\sup_{t\in\mathbb{R}}\left|\left<e^{itL}g,g^{*}\right>\right|.
$$
Since
$$
\left(\frac{d}{dt}\right)^{k}\left<e^{itL}g,g^{*}\right>=
\left<e^{itL}L^{k}g,g^{*}\right>
$$
we obtain for $t=0$
$$
\left|\left<L^{k}g,g^{*}\right>\right|\leq
\omega^{k}\|g^{*}\|\|f\|\int_{-\infty}^{\infty}|h(\tau)| d\tau.
$$
 Choosing $g^{*}$ such that $\|g^{*}\|=1$ and $\left<L^{k}g,g^{*}\right>=
 \|L^{k}g\|$
  we obtain the inequality
\begin{equation}
\|L^{k}g\|\leq
\omega^{k}\|f\|\int_{-\infty}^{\infty}|h(\tau)|d\tau\equiv
C_{h}\omega^{k}\|f\|,\label{L}
\end{equation}
where
$$
 C_{h}=\int_{-\infty}^{\infty}|h(\tau)|d\tau.
$$
 Now we make
an important observation that regardless of the value of $C_{h}$
the inequality (\ref{L}) implies that $g$ belongs to
$PW_{\omega}(L)$. Indeed, for any complex number $z$ we have

$$ \|e^{izL}g\|=\left\|\sum ^{\infty}_{k=0}(i^{k}z^{k}L^{k}g)/k!\right\|
\leq C_{h}\|f\| \sum
^{\infty}_{k=0}|z|^{k}\omega^{k}/k!=C_{h}\|f\|e^{|z|\omega}.
$$ It
implies that for any $g^{*}\in L_{2}(G)$ the scalar function
$<e^{izL}g,g^{*}>$ is an entire function
 of exponential type $\omega $ which is bounded on the real axis $R^{1}$
by the constant $\|g^{*}\| \|g\|$.
 An application of the Bernstein inequality gives

$$\|<e^{itL}L^{k}g,g^{*}>\|_{C(R^{1})}=
\left\|\left(\frac{d}{dt}\right)^{k}<e^{itL}g,g^{*}>\right\|_{
C(R^{1})} \leq\omega^{k}\|g^{*}\| \|g\|.
$$ The last one gives for
$t=0$
$$ \left|<L^{k}g,g^{*}>\right|\leq \omega ^{k} \|g^{*}\| \|g\|.
$$
Choosing $g^{*}$ such that $\|g^{*}\|=1$ and
$<L^{k}g,g^{*}>=\|L^{k}g\|$ we obtain the inequality
$\|L^{k}g\|\leq
 \omega ^{k} \|g\|, k\in N$. The Lemma is proved.

\end{proof}

We will also need the following Lemma.
\begin{lem}
The following inequalities  hold for all $f\in L_{2}(G)$
\begin{equation}
\Omega_{m}\left(f, s\right)\leq
s^{k}\Omega_{m-k}(L^{k}f,s), \  \  0\leq k\leq m,  \label{mod.cont.1}
\end{equation}
 
 and 
 
\begin{equation}
\Omega_{m}\left(f, ns\right)\leq n^{m}\Omega_{m}(f,s), \   \  n, m\in \mathbb{N}.
\end{equation}
\label{tech}
\end{lem}
\begin{proof}
The following identity holds
$$
\left(e^{itL}-I\right)f=i\int_{0}^{t}e^{i\tau
 L}Lfd\tau,
$$
where $I$ is the identity operator.  Iterations of this formula give the identity
$$
\left(e^{itL}-I\right)^{k}f=i^{k}\int_{0}^{t}...\int_{0}^{t}e^{i(\tau_{1}+...+\tau_{k})L} L^{k}f d\tau_{1}...d\tau_{k},
$$
which implies (4.4).

The
second one follows from the property
$$
\Omega_{1}\left(f, s_{1}+s_{2}\right)\leq \Omega_{1}\left(f,
s_{1}\right)+\Omega_{1}\left(f,s_{2}\right),
$$
which is easy to verify.

\end{proof}

Bellow  the following function will be used
\begin{equation}
h(t)=a\left(\frac{\sin (t/n)}{t}\right)^{n},\label{functionh}
\end{equation}
where $n$ is a fixed  even integer and
\begin{equation}
a=\left(\int_{-\infty}^{\infty}\left(\frac{\sin
(t/n)}{t}\right)^{n}dt\right)^{-1}.\label{constanta}
\end{equation}
Now we construct another filtering  operator

$$
\mathcal{P}_{h}^{\omega,m}:L_{2}(G)\rightarrow PW_{\omega}(L),
$$
 which is defined as
\begin{equation}
\mathcal{P}_{h}^{\omega,m}(f)=\int_{-\infty}^{\infty}
h(t)\left\{(-1)^{m-1}\Delta^{m}_{t/\omega}f+f\right\}dt,\label{id}
 \end{equation}
 where
\begin{equation}
(-1)^{m+1}\Delta^{m}_{s}f=\sum^{m}_{j=0}(-1)^{j-1}C^{j}_{m}e^{js(iL)}f=
 \sum_{j=1}^{m}b_{j}e^{js(iL)}f-f,\label{id2}
 \end{equation}
 and
 \begin{equation}
 b_{1}+b_{2}+...
+b_{m}=1.\label{id3}
\end{equation}
  The next Theorem is an analog of the
classical Direct Approximation Theorem by entire functions of
exponential type.

\begin{thm} Let $h$ be the function defined in (\ref{functionh}) and
(\ref{constanta}) and the operator $\mathcal{P}_{h}^{\omega,m}$ is
defined in (\ref{id})-(\ref{id3}). We also assume that the
following inequality holds
\begin{equation}
n\geq m+2.\label{nm-ineq}
 \end{equation}
 For any appropriate  $n\in \mathbb{N}$ and $m\in \mathbb{N}$
 that satisfy (\ref{nm-ineq}) and for
  every natural $k$ such that
  $$
  0\leq k\leq m, k,m\in
\mathbb{N},
$$
there exists a constant $C^{h}_{k,m}>0$ such that for all $\omega
> 0$ and all
 $f\in L_{2}(G)$ the following inequalities  holds
\begin{equation}
E(f,\omega)\leq
\|\mathcal{P}_{h}^{\omega,m}(f)-f\|\leq\frac{C^{h}_{m,k}}{\omega^{k}}
\Omega_{m-k}\left(L^{k}f, 1/\omega\right), \label{J1}
\end{equation}\label{JT}
 where
$$
C^{h}_{m,k}=\int_{-\infty}^{\infty}h(t)|t|^{k}(1+|t|)^{m-k}dt, 0\leq
k\leq m.
$$
\label{DAT}
\end{thm}
\begin{proof}

With the  choice of $a$ as in (\ref{constanta})  and $ n\geq m+2$
 the function $h$ will have the
 following properties: 1) $h$ is an even nonnegative entire function of exponential type
one; 2) $h$ belongs to $L_{1}(\mathbb{R})$ and its
$L_{1}(\mathbb{R})$-norm is $1$;  3) the integral
\begin{equation}\int_{-\infty}^{\infty}h(t)|t|^{m}dt
\end{equation} is finite.
The formulas (\ref{id}) and (\ref{id2}) imply the next formula
$$
\int_{-\infty}^{\infty}h(t)\sum_{j=1}^{m}b_{j}e^{j\frac{t}{\omega}(iL)
} fdt=\int_{-\infty}^{\infty}\Phi(t)e^{t(iL)}fdt.
$$
where
$$
\Phi(t)=\sum_{j=1}^{m}b_{j}\left(\frac{\omega}{j}\right)h\left(t\frac{\omega}{j}\right).
$$
Since the function $h(t)$ has exponential type one, every function
$h(t\omega/j)$ has the type $\omega/j$ and because of this the
function $\Phi(t)$ has exponential  type $\omega$.

Now we estimate the error of approximation of
$\mathcal{P}_{h}^{\omega,m}(f)$  to $f$.  Since by (\ref{id})
$$
f-\mathcal{P}_{h}^{\omega,m}(f)=(-1)^{m}
\int_{-\infty}^{\infty}h(t)\Delta^{m}_{t/\omega}fdt
$$
we obtain
$$
E(f,\omega)\leq\|f-\mathcal{P}_{h}^{\omega,m}(f)\|\leq
\int_{-\infty}^{\infty}h(t)\left\|\Delta^{m}_{t/\omega}f\right\|dt\leq
\int_{-\infty}^{\infty}h(t)\Omega_{m}\left(f, t/\omega\right)dt.
$$
By using the Lemma \ref{tech} we obtain
$$
E(f,\omega)\leq \int_{-\infty}^{\infty}h(t)\Omega_{m}\left(f,
t/\omega\right)dt \leq \frac{\Omega_{m-k}\left(L^{k}f,
1/\omega\right)}{\omega^{k}}\int_{-\infty}^{\infty}h(t)|t|^{k}(1+|t|)^{m-k}dt\leq
$$
$$
\frac{{C}^{h}_{m,k}}{\omega^{k}}\Omega_{m-k}\left(L^{k}f,
1/\omega\right),
$$
where the integral
 $$
C^{h}_{m,k}=\int_{-\infty}^{\infty}h(t)|t|^{k}(1+|t|)^{m-k}dt
$$
 is finite by the choice of $h$. The inequality (\ref{J1}) is proved.
\end{proof}

Now we can formulate one of our main results about \textbf{sparse
approximation} of functions in $L_{2}(G)$.
 Namely, a combination of the Theorem \ref{SamplingTheorem} with
the Theorem \ref{JT} gives the following result about
approximation of an $f\in L_{2}(G)$  by using samples of its
orthogonal projection $f_{\omega}$ on $PW_{\omega}(L)$ or samples
of the projection $\mathcal{P}_{h}^{\omega,m}(f)$.

In the following Theorem we assume that $h$ is the function
defined in (\ref{functionh}) and (\ref{constanta}) and the
operator $\mathcal{P}_{h}^{\omega,m}$ is defined in
(\ref{id})-(\ref{id3}). Let us also recall that in the Theorem
\ref{SamplingTheorem} a subset of vertices $S$ is  such that every
vertex $v$ in $bS$ is connected to at least $K_{0}(S)$ vertices in
$S$, where $1\leq K_{0}(S)\leq d(v)$.

The next Theorem represents our result about sparse approximation
on graphs.

\begin{col} If $S\subset V(G)$ is the same as in the Theorem  \ref{SamplingTheorem}
and the condition (\ref{cl1}) along with inequality
$$
\omega<\frac{K_{0}(S)}{4}
$$
hold, then there exist functions $\Theta_{u}\in PW_{\omega}(L),
u\in S,$ and for
   any $0\leq k\leq m, \>\>\>\> k,m \in \mathbb{N},$ there exists a constant $C_{k,m}>0$ such
   that  for every $f\in L_{2}(G)$
\begin{equation}
\left \|f-\sum_{u\in S}f_{\omega}(u)\Theta_{u}\right\|\leq \left
\|f-\sum_{u\in
S}\mathcal{P}_{h}^{\omega,m}f(u)\Theta_{u}\right\|\leq
\frac{C_{k,m}}{\omega^{k}}\Omega_{m-k}\left(L^{k}f,
1/\omega\right),
\end{equation}
where $f_{\omega}$ is the orthogonal projection of $f$ on
$PW_{\omega}(L)$. \label{MAT}
\end{col}

\section{ Lattice $\mathbb{Z}^{n}$}.

The Fourier transform $\mathcal{F}$ on $L_{2}(\mathbb{Z}^{n})$ is
a unitary  operator
$$\mathcal{F}:L_{2}(\mathbb{Z}^{n})\rightarrow
L_{2}\left(\mathbb{T}^{n}, d\xi/(2\pi)^{n}\right),
$$
where $\mathbb{T}^{n}$ is the n-dimensional torus and
$d\xi/(2\pi)^{n}$ is the normalized measure which is defined by
the formula
$$
\mathcal{F}(f)(\xi_{1},
\xi_{2},...,\xi_{n})=\sum_{(k_{1},k_{2},...,k_{n})\in
\mathbb{Z}^{n}}f(k_{1}, k_{2},...,k_{n})e^{i k_{1}\xi_{1}+i
k_{2}\xi_{2}+...+i k_{n}\xi_{n}},
$$
where $  f\in
L_{2}(\mathbb{Z}^{n}),(\xi_{1},\xi_{2},...,\xi_{n})\in [-\pi,
\pi)^{n}$ .  One can verify the following formula
$$
\mathcal{F}(Lf)(\xi)=4\left(\sin^{2}\frac{\xi_{1}}{2}+
\sin^{2}\frac{\xi_{2}}{2}+...+\sin^{2}\frac{\xi_{n}}{2}\right)\mathcal{F}(f)(\xi),
$$
where $  f\in
L_{2}(\mathbb{Z}^{n}),(\xi_{1},\xi_{2},...,\xi_{n})\in [-\pi,
\pi)^{n}$.
\begin{thm} \cite{Pes3} The spectrum of the Laplace operator $L$
 on the
lattice $\mathbb{Z}^{n}$ is the set $[0, 4n]$. A function $f$
belongs to a space $PW_{\omega}(\mathbb{Z}^{n})$ for some  $
0<\omega<4n,$ if and only if the support of $\mathcal{F}f$ is a
subset $\Omega_{\omega}$ of $[-\pi,\pi)^{n}$ on which
$$
\sin^{2}\frac{\xi_{1}}{2}+
\sin^{2}\frac{\xi_{2}}{2}+...+\sin^{2}\frac{\xi_{n}}{2}\leq
\frac{\omega}{4}.
$$

\end{thm}

Let's consider for simplicity the case
$\mathbb{Z}^{2}=\{(k_{1},k_{2})\}, k_{1}, k_{2}\in \mathbb{Z}$. We
will use notation
$$
S=\mathbb{Z}^{2}\setminus
\left\{(2k_{1},2k_{2})\right\}_{k_{1},k_{2}\in \mathbb{Z}}
$$
for the  set of vertices which is  a compliment of the set of
vertices $\{(2k_{1},2k_{2})\}, k_{1},k_{2}\in \mathbb{Z}$.  For
this set the assumptions of the Sampling Theorem
\ref{SamplingTheorem} will be satisfied with $K_{0}(S)=4$ and we
obtain the following fact.
\begin{thm}
If $ \omega<1, $
 then

1) the set $S$ is a uniqueness set for functions in
$PW_{\omega}(L)$;

 2) there exist functions $\Theta_{u}\in PW_{\omega}(L)$, $u\in S,$
such that for all $f\in PW_{\omega}(L)$ the following
reconstruction formula holds
\begin{equation}
f=\sum_{u\in S}f(u)\Theta_{u}.
\end{equation}
\end{thm}

If $S$ is the set $\{(k_{1}, 2k_{2})\}, k_{1},k_{2}\in
\mathbb{Z}$, then the number $k$ in the Theorem
\ref{SamplingTheorem} is $K_{0}(S)=2$ and we have a  similar
result for all spaces $PW_{\omega}(L)$ with $\omega<1/2.$

If $S$ is the set $\{(k_{1}, 3k_{2})\}, k_{1},k_{2}\in
\mathbb{Z}$, then the number $K_{0}(S)$ in the Theorem
\ref{SamplingTheorem} is $K_{0}(S)=1$ and we have a result similar
to the last Theorem for all spaces $PW_{\omega}(L)$ with
$\omega<1/4.$

To construct a projector
$$
\mathcal{P}_{h}^{\omega,m}:L_{2}(\mathbb{Z}^{2})\rightarrow
PW_{\omega}(L)
$$
by the formula (\ref{id}) one can use the following description of
the Schrodinger's group of operators
$$
e^{itL}f=\mathcal{F}^{-1}\left(e^{-4it\left(\sin^{2}\frac{\xi_{1}}{2}+
\sin^{2}\frac{\xi_{2}}{2}\right)}\right)\mathcal{F}f.
$$

Now one can easily reformulate all statements  of the previous
section for the lattice $\mathbb{Z}^{2}$ and each of  the sets $S$
 above.

Our results  not only  alow to reduce analysis on the lattice
$\mathbb{Z}^{2}$ to analysis on appropriate subgraph, but they
also give estimates of possible losses   of information which can
occur after such reduction.

\section{Acknowledgment}

Authors would like to thank the anonymous referee for useful and constructive
suggestions.

\end{document}